\documentclass[a4paper,11pt]{amsart}

\usepackage[latin1]{inputenc}
\usepackage{amsmath}
\usepackage{amsthm}
\usepackage{amscd}
\usepackage{amssymb}
\usepackage{amsfonts}
\usepackage{hyperref}

\newcommand{\calC}{\mathcal{C}}
\newcommand{\F}{\mathbb{F}}

\newtheorem*{thm*}{Theorem}
\newtheorem{thm}{Theorem}[section]

\newtheorem{lem}[thm]{Lemma}
\newtheorem{con}[thm]{Conjecture}
\newtheorem{coro}[thm]{Corollary}
\newtheorem{rmk}[thm]{Remark}

\begin{document}

\title{Asymptotic performance of metacyclic codes}
\author{Martino Borello}
\thanks{M. Borello is with LAGA,  UMR 7539, CNRS, Universit\'e Paris 13 - Sorbonne Paris
Cit\'e, Universit\'e Paris 8, F-93526, Saint-Denis, France}
\author{Pieter Moree}
\thanks{P. Moree is with Max-Planck-Institut f\"ur Mathematik, Vivatsgasse 7, D-53111 Bonn,
Germany.}
\author{Patrick Sol\'e}
\thanks{P. Sol\'e is with Aix Marseille University, CNRS, Centrale Marseille, I2M, Marseille, France}

\date{}
\maketitle

\begin{abstract} 
A finite group with a cyclic normal subgroup $N$ such
that $G/N$ is cyclic is said to be metacyclic. 
A code over a finite field $\F$ is a metacyclic code if it is a left ideal in the group algebra $\F G$ for $G$ a metacyclic group. Metacyclic codes are generalizations of dihedral
codes, and can be constructed as quasi-cyclic codes with an extra
automorphism. In this paper, we prove that metacyclic codes form an
asymptotically good family of codes. Our proof relies on a version of
Artin's conjecture for primitive roots in arithmetic progression being true
under the Generalized Riemann Hypothesis (GRH).
\end{abstract}

\bigskip


\section{Introduction}
Metacyclic codes were studied intensively by Sabin in the 1990's
\cite{S,SL}. They are (left) ideals in the group ring $\F G(m,s,r),$
where $\F$ is a finite field and $G(m,s,r)$ is the finite group of order
$ms$ defined as
$$G(m,s,r)=\langle x,y \mid x^m=1,\,y^s=1,\, yx=x^ry\rangle,$$
with $r^s\equiv 1\,({\rm mod~}m)$. More recently, their concatenated
structure was explored in \cite{CCFG}.

In the present paper, we show that for some values of the parameters
$m,s,r$ (in particular $s>1$ fixed, $m$ a prime and $r\neq 1\,({\rm mod~}m)$ depending on
$m$), these codes are asymptotically good. This extends results of
Bazzi-Mitter, who dealt with $\F=\F_2$ and $G(m,2,m-1),$ a dihedral
group \cite{BM}, and, partially, Borello-Willems, who considered the
case $\F=\F_s,$ with both $m$ and $s$ prime. As observed in
\cite[\S4]{BCW}, applying field extensions as in \cite[Proposition
12]{FW}, the result of Bazzi-Mitter can be extended to any field of
characteristic $2$ and that of Borello-Willems to any field
of characteristic $m$. In our case, the characteristic of the field
is not necessarily related to the cardinality of the group, so that
we have more freedom in our choice of the alphabet. Moreover, the
proof is conceptually simpler. On the other hand, our results rely on
a variant of Artin's primitive conjecture (Conjecture \ref{generalArtin}) being true, where
the primes are supposed to lie in a progression of the form 
$1\,({\rm mod~}s).$ This is currently only guaranteed on
assuming the GRH.

The main idea is to realize metacyclic codes as quasi-cyclic codes
with some extra automorphism. Such codes can be enumerated by the Chinese Remainder Theorem (CRT) approach of \cite{LS1}. This technique regards a
quasi-cyclic code of index $\ell$ as a code of length $\ell$ over an
auxiliary ring. Decomposing the said ring
into a direct sum of extension fields by the CRT for polynomials
yields a decomposition into codes of length $\ell$ over these fields. These codes are called constituent codes. The
favorable case where there are only two constituents requires, to be realized
infinitely many times, to invoke  Artin's conjecture. An
expurgated random coding argument, similar to the one that proves
that double circulant codes are asymptotically good \cite{A+}, can
then be applied.

The material is arranged as follows. The next section collects the basic notions and notations needed in the rest of the paper. Section \ref{main-result} derives the main result. Section \ref{conclusion} concludes the article.
\section{Definitions and Notation}
\subsection{Quasi-cyclic Codes}
A linear code $\calC$ over the finite field $\F$ is said to be {\bf
quasi-cyclic} of index $\ell$, or $\ell$-QC for short, if it is left
wholly invariant under the $\ell$'th power of the shift. Assume, for
convenience, that the length $n$ of $\calC$ is $n=\ell m,$ for some
integer $m$ called the co-index. As is well-known \cite{LS1}, such a
code is an $R_m$-submodule of $R_m^\ell,$ where $R_m$ denotes the
ring $R_m=\F[x]/(x^m-1).$ A related class of codes is that of {\bf
$\ell$-circulant} codes, consisting of linear codes of length
$n=\ell m$ and dimension $m,$ whose generator matrix is made of
circulant blocks of size $m$. Such
codes are coordinate permutation equivalent to $\ell$-QC codes.
Recall that there is ring isomorphism between circulant matrices of order $m$
and $R_m$ given by $A \mapsto A_{11}+A_{12}x+\cdots A_{1m}x^{m-1}.$
Thus, for example, the binary matrix $$\begin{pmatrix} 1 & 0 & 0 & 0 & 1 & 1\\
                                         0 & 1 & 0 & 1 & 0 & 1\\
                                         0 & 0 &1 & 1 & 1 & 0
 \end{pmatrix}
$$ is encoded by that isomorphism as $(1,x+x^2).$
\subsection{Metacyclic groups}
Let $G(m,s,r)$ denote the group of order $ms$ defined by generators
and relations as
$$G(m,s,r)=\langle x,y \mid x^m=1,\,y^s=1,\, yx=x^ry\rangle,  $$
where $r$ satisfies $r^s\equiv 1 \pmod*{m}.$ Such a group is called
{\bf metacyclic}, since it has a cyclic normal subgroup $N=\langle
x\rangle$ such that the quotient group $G/N$ is also cyclic. When
$r=m-1$ and $s=2,$ we obtain the {\bf dihedral} group $D_m$ of order
$2m$,
$$D_m=\langle x,y \mid x^m=1,\,y^2=1,\, yx=x^{-1}y\rangle,  $$
while the case $r=1$ reduces to the abelian group
$$C_m\times C_s=\langle x,y \mid x^m=1,\,y^s=1,\, yx=xy\rangle. $$
Here $C_i$ denotes the cyclic group of order $i.$
\subsection{Group codes}
Let $G$ be a finite group of order $n$. A {\bf $G$-code} (or a {\bf
group code}) $\calC$ over a finite field $\F$ is a left ideal in the
group algebra $\F G=\{\sum_{g\in G}a_gg\mid a_g\in \F\}$. Once 
we choose an ordering of $G$, we have a $\F$-linear isomorphism
$\varphi:\sum_{g\in G}a_gg\mapsto(a_g)_{g\in G}$ between $\F G$ and
$\F^n$, and the image of $\calC$ is a linear code in $\F^n$.
Changing the ordering gives coordinate permutation equivalent codes.
The group of permutation automorphism of $\varphi(\calC)$ contains a
transitive subgroup isomorphic to $G$. It is common practice to
identify $\calC$ and $\varphi(\calC)$. A {\bf metacyclic code} is a
$G$-code for $G=G(m,s,r)$ or equivalently a linear code of length
$ms$ whose permutation automorphism group contains a transitive
subgroup isomorphic to $G(m,s,r)$.
\subsection{The Artin primitive root conjecture for primes in 
arithmetic progression}\label{Artin}
Emil Artin conjectured in 1927 that given a non-zero integer $a$
that is not a perfect square nor $-1$, there are infinitely many
primes $m$ such that $a$ is primitive modulo $m.$ Recall that the Generalized Riemann Hypothesis (GRH) states that the analogue of Riemann hypothesis for zeta functions of number fields \cite{D}, the so 
called Dedekind zeta functions, holds true. 
A quantitative version of Artin's primitive root
conjecture is
proved under GRH by Hooley \cite{H}, and unconditionally for all but
two unspecified prime roots $a$ by Heath-Brown \cite{He}. 
The following is a refinement of Artin's primitive root conjecture
where in addition the prime $m$ is required to 
be in a fixed arithmetic progression 
$1\,({\rm mod~}s).$ 
\begin{con} 
\label{generalArtin}
Let $a$ be a non-zero integer that is not a perfect square nor $-1$.
Let $h$ be the largest integer such that $a$ is an $h$-th power. 
Let
$\Delta$ denote the discriminant of $\mathbb Q(\sqrt{a})$.
Given $s\ge 1,$ let $S(a,s)$ be the set of primes $m\equiv
1\,({\rm mod~}s)$ such that $a$ is a primitive root modulo $m.$ 
If both $(s,h)=1$ and $\Delta\nmid s,$ then
the set $S(a,s)$ is infinite.
\end{con} 
If $(s,h)>1,$ then the set
$S(a,s)$ is finite. Suppose there exists an element $m$ of that set not dividing $a$. Writing $a=a_0^h,$ we have 
$a^{(m-1)/(s,h)}=a_0^{h(m-1)/(s,h)}\equiv a_0^{m-1}\equiv 1\,({\rm mod~}m)$ and
so $a$ is not primitive modulo $m.$ Contradiction.

Likewise, if $\Delta\mid s,$ the set $S(a,s)$ is finite.  
By elementary algebraic number theory
the smallest $m$ for which $\mathbb Q(\sqrt{g})\subseteq \mathbb Q(\zeta_k)$ equals
$k=|\Delta|.$ The primes $m\equiv 1\,({\rm mod~}s)$ split
completely in $\mathbb Q(\zeta_k)$ and so certainly in the subfield $\mathbb Q(\sqrt{a})$.
If $m\nmid a,$ it then follows that the Legendre symbol $(a/m)=1$ and so the order of $a$ modulo $m$
is at most $(m-1)/2,$ and so  $S(a,s)$ is finite.  

 The conjecture thus claims that if
there is no trivial reason for $S(a,s)$ to be finite, it is actually
infinite.

Under GRH Lenstra \cite[Theorem 8.3]{L} established a far reaching
generalization of Artin's original conjecture. In particular, his 
work implies the truth of Conjecture \ref{generalArtin}. 
Indeed, under GRH the set $S(a,s)$ has a
natural density that can be explicitly given, which was done by
Moree \cite[Theorem 4]{M}. Combination of the two results yields the following theorem.
\begin{thm} 
\label{Artinprogression}
Under GRH Conjecture  \ref{generalArtin} holds true and, moreover,  
the set $S(m,s)$ has an explicitly
determinable density that is a rational multiple times the
Artin constant.
\end{thm}
The reader interested in more information regarding the Artin primitive
root conjecture and its many generalizations and applications is referred
to the survey \cite{Survey}.

\subsection{Asymptotics}\label{asymp}
If $\calC(n)$ is a family of codes with parameters $[n,k_n,d_n]$
over $\F_q$, the rate $R$ and relative distance $\delta$ are defined
as $$R=\limsup\limits_{n\rightarrow \infty}\frac{k_n}{n}\text{~and~}
\delta=\liminf\limits_{n\rightarrow \infty}\frac{d_n}{n},$$
respectively. When examining a family of codes, it is natural to ask
if this family is asymptotically good or bad in the following sense.
A family of code is {\bf asymptotically good} if $R\delta \neq 0.$

Recall the $q$-ary {\bf entropy function} defined for $0\leq t\leq\frac{q-1}{q}$ by
\begin{equation*}\label{den1}
H_q(t)=\begin{cases}
 \emph{ }0,  ~~~~~~~~~~~~~~~~~~~~~~~~~~~~~~~~~~~~~~~~~~~~~~~~~~~~~~~~~~~~~~~~{\rm{if}}~~ t=0,\\
   \emph{ }t{\rm{log}}_q(q-1)-t{\rm{log}}_q(t)-(1-t){\rm{log}}_q(1-t), ~~~~~~~~~~~{\rm{if}}~~0<t\leq\frac{q-1}{q}. \\
\end{cases}
\end{equation*}
This quantity is instrumental in the estimation of the volume of high-dimensional Hamming balls when the base field is $\mathbb{F}_q$.
The result we are using in this paper is that the volume of the Hamming ball of radius $tn$ is asymptotically equivalent, up to subexponential terms, to $q^{nH_q(t)}$, when
$0<t<1$, and $n$ goes to infinity \cite[Lemma 2.10.3]{HP}.
\section{Main result}\label{main-result}
 Let $s$ be an integer greater than $1$ and $T_{a_1,\ldots,a_{s-1}}$ denote the $s$-circulant code over $\F_q$ with generator matrix $(1,a_1(x),\ldots,a_{s-1}(x))$ with $a_i(x)\in
 R_m=\F_q[x]/(x^m-1).$
 Denote by $\mu_r$ the {\bf multiplier by $r$} in $R_m,$ defined for all $f(x)\in R_m$ by $\mu_r(f(x))=f(x^r).$ (Cf. \cite[\S 4.3]{HP}).
Note that
$$\F_qG(m,s,r)\simeq \F_q[x,y]/(x^m-1,y^s-1,x^ry-yx)$$
as a ring (we are just choosing a special ordering of the elements
of $G(m,s,r)$), and the right-hand side is isomorphic to $R_m^s$ as an
$R_m$-module via
$$f_1(x)+f_2(x)y+\ldots+f_s(x)y^{s-1}\mapsto(f_1(x),f_2(x),\ldots,f_s(x)).$$
A construction of metacyclic codes from quasi-cyclic codes similar
to the next lemma can be found in \cite[Theorem 1]{S}.
 \begin{lem} 
 \label{ideal} 
 If $a_1\mu_r(a_1)\ldots\mu^{s-1}_r(a_1)=1$ and $a_j=a_1\mu_r(a_1)...\mu_r^{j-1}(a_1)$ for all $j\in\{2,\ldots,s-1\}$, then $T_{a_1,\ldots,a_{s-1}}$
 is metacyclic for the group $G(m,s,r).$
 \end{lem}
 \begin{proof}
  Writing an arbitrary codeword as $$f_1(x)+f_2(x)y+\ldots+f_s(x)y^{s-1} \in \F_q[x,y]/(x^m-1,y^s-1,x^ry-yx),$$
  we see that left multiplication by $y$ in that ring corresponds to the map
  $$(f_1(x),\ldots,f_s(x))\mapsto(\mu_r(f_s(x)),\mu_r(f_1(x)),\ldots,\mu_r(f_{s-1}(x)))$$
in $R_m^s$. For $T_{a_1,\ldots,a_{s-1}}$ to be a $G(m,s,r)$-code, it
is sufficient to check that it is stable under left multiplication by $y$
(every $s$-circulant code being clearly stable under left
multiplication by $x$). Reasoning on the generator of
$T_{a_1,\ldots,a_{s}}$ the above relation shows that
$(\mu_r(a_{s-1}(x)),1,\ldots,$ $\mu_r(a_{s-2}(x)))$ is proportional to
$(1,a_1(x),\ldots,a_{s-1}(x))$ by an element of $R_m.$
  Getting rid of that element between two equations yields the said relations on the elements $a_1(x),\ldots,a_{s-1}(x).$
 \end{proof}
 
 \begin{rmk}\label{rmk-abelian} {\rm
A natural question is under which conditions $T_{a_1,\ldots,a_{s-1}}$ is two-sided, since in this case the code would be abelian \cite{SL}. Reasoning in the same way as above on the right multiplication by $y$ and by $x$ we obtain that $T_{a_1,\ldots,a_{s-1}}$ is a right ideal in $\F_q[x,y]/(x^m-1,y^s-1,x^ry-yx)$ if and only if $a_{s-1}^s=1$, $a_j=a_{s-1}^{s-j}$ for all
$j\in\{1,\ldots,s-2\}$, and $a_j=a_j\cdot x^{jr-1}$ for all $j\in
\{1,\ldots,s-1\}$. 
The last condition is equivalent to $a_j$ being constant on the orbits of the $(jr-1)$-th power of the shift, and in the 
case $m$ is prime and $r\neq 1 \,({\rm mod~}m)$, it is easy to see that the set of $T_{a_1,\ldots,a_{s-1}}$ satisfying all above conditions is empty: actually $a_1=a_1\cdot x^{r-1}$ implies 
that $a_1=\lambda(1+\ldots+x^{m-1})$, with $\lambda\in \F_q$. But then $$a_1\mu_r(a_1)\ldots\mu^{s-1}_r(a_1)=\lambda^s(1+\ldots+x^{m-1})^s
=\lambda^sm^{s-1}(1+\ldots+x^{m-1})$$ (the last equality can be proven by induction on $m\geq 2$) cannot be equal to $1$.}
 \end{rmk}

 We now assume that $m$ is a prime, such that $q$ is primitive modulo $m$. Thus, by the theory of cyclotomic cosets \cite[\S 4.1]{HP}, we know that $x^m-1=(x-1)h(x),$ with $h$ irreducible over $\F_q[x].$

 \begin{lem} 
 \label{count}
 Assume that $s$ divides $m-1$ and that the order of $r$ modulo $m$ is $s.$
 The number $\Omega_{m,s}$ of the $s$-circulant codes with the properties in Lemma \ref{ideal} is 
 $$\Omega_{m,s}=s'\cdot \frac{q^{m-1}-1}{q^{\frac{m-1}{s}}-1},
 \text{\,\,with\,\,} s'=(s,q-1).$$
 \end{lem}
\begin{proof}
  The CRT for polynomials yields the ring decomposition $$R_m\simeq \F_q\oplus \F_Q,$$ with $Q=q^{m-1}.$

  Write $a_1=a_1' \oplus \alpha_1, \ldots, a_{s-1}=a_{s-1}' \oplus \alpha_{s-1},$ in this decomposition. We study the conditions on $a_1,\ldots,a_{s-1}$
  given in Lemma \ref{ideal}, in the light of this CRT decomposition.
  \begin{itemize}
   \item The conditions on $a_1',\ldots,a_{s-1}'$ are $a_1'^s=1,$ and $a_j'=a_1'^{j-1}$ for $j\in\{2,\ldots,s-1\}$. The first equation has $s'$ solutions, with
   $s'=(s,q-1)$ and the rest of $a_j$ is uniquely determined.
   \item Since, by hypothesis, the order of $\mu_r$ is $s,$ the action of $\mu_r$ on $\F_Q,$ by the characterization of the Galois group
   of $\F_Q$ is exponentiation by $t=q^{\frac{m-1}{s}}.$ The condition on $a_1$ implies
   $$\alpha_1\in \{ z \in \F_Q \mid z^{1+t+\ldots+t^{s-1}}=1\}= \{ A^{t-1} \mid A \in \F_Q^\times\}, $$
   a set of size $\frac{t^s-1}{t-1}.$
   The rest of the $\alpha_j$'s is uniquely determined.
  \end{itemize}
The result follows by multiplying these two independent counts together.

   \end{proof}
   
The next lemma shows that the codes of Lemma \ref{ideal} have ``small'' common intersection.
\begin{lem} 
 \label{cover} If $(f_1(x),\ldots,f_s(x)) \in R_m^s,$ with a Hamming weight $<m$, then there are at most
 $q$ codes $T_{a_1,\ldots,a_{s-1}}$ with the properties in Lemma \ref{ideal} such that $(f_1(x),\ldots,f_s(x)) \in T_{a_1,\ldots,a_{s-1}}.$
 \end{lem}
\begin{proof}
 Keep the notation of the proof of Lemma \ref{count}. Since $a_2,\ldots,a_{s-1}$ are uniquely determined  by $a_1,$ we focus on $a_1.$
 The Hamming weight condition implies that $f_1(x)  \neq 0 \,({\rm mod~}h(x))$ (otherwise $f_1$ would be a nonzero codeword of the repetion code of length $m$
 over $\F_q$).
  Then $a_1(x)$ is uniquely determined modulo $h(x)$ by the equation $f_2(x)\equiv f_1(x)a_1(x)\,({\rm mod~}h(x)).$ But modulo $x-1$ it can takes $q$ values.
   The result follows.
 \end{proof}

The following results are true under Artin's primitive root
conjecture Conjecture \ref{generalArtin} for the progression 
$1\,({\rm mod~}s),$ which by Theorem \ref{Artinprogression} is 
guaranteed if GRH holds true.
\begin{thm} 
\label{main} Assume Conjecture \ref{generalArtin} holds true. Let $q$ be a prime and $s>1$ be an integer such that $q\nmid s$ if $q\equiv 1\,({\rm mod~}4)$ and $4q\nmid s$ if $q\equiv 3\,({\rm mod~}4)$ or $q=2$.
 For every $0<\delta <H_q^{-1}(\frac{s-1}{s^2}),$ there is a sequence of metacyclic codes
 over $\F_q$ that are group codes for $G(m,s,r)$ of rate
$1/s$ and
 relative Hamming distance $\delta.$
\end{thm}
 \begin{proof}
 Under these hypotheses on $q$ and $s$, the existence of infinitely many primes $m$ such that $q$ is primitive modulo $m$
 and that $s$ divides $m-1$ is ensured by Artin's primitive root conjecture in arithmetic progression, as shown in \S\ref{Artin} (for the discriminant
 of quadratic fields see \cite[p.89]{Sam}).
  If the number $\Omega_{m,s}$ is strictly larger than $q$ times the size of a Hamming ball of radius $\lfloor \delta ms \rfloor,$ then, by Lemma \ref{cover},
  there is a code constructed by Lemma \ref{ideal} of minimum distance $>\lfloor \delta ms \rfloor.$
  This inequality will hold if, using the standard entropic estimates of \S \ref{asymp}, for $m \to \infty,$ we have
  $$ s'\cdot \frac{q^{m-1}-1}{q^{\frac{m-1}{s}}-1}>q\cdot q^{msH_q(\delta)},$$ and in particular if $(s-1)/s^2>H_q(\delta).$
 \end{proof}

 We relax the condition that $q$ is prime as follows.

 \begin{coro} 
 \label{asym} 
 Assume Conjecture \ref{generalArtin} holds true. 
 The metacyclic codes over $\F_w$ 
 with $w=q^a,$ $q$ a prime and $a\ge 2,$ form an asymptotically good family of codes.
\end{coro}
 \begin{proof}
  By the preceding theorem the metacyclic codes over $\F_q$ are asymptotically good. By extension of scalars from $\F_q$ to $\F_w$ (as in \cite[Proposition
12]{FW}) the result
  follows.
 \end{proof}

 The following result was proved for $q$ even by similar techniques, under Artin's conjecture, in \cite{A+}, and unconditionally in \cite{BM} using more advanced  probabilistic techniques.

 \begin{coro} 
 Assume Conjecture \ref{generalArtin} holds true. 
 Dihedral codes are asymptotically good in any characteristic.
 \end{coro}
\begin{proof}
  A consequence of Theorem \ref{main} and Corollary \ref{asym} in the case $r=m-1$ and $s=2$, for which $G(m,s,r)$ is the dihedral group of order $2m.$
 \end{proof}

 The following result was proved unconditionally and for all characteristics in \cite{BW} using the Bazzi-Mitter approach of \cite{BM}.

 \begin{coro} 
 Assume Conjecture \ref{generalArtin} holds true. 
 If $p\equiv 3\,({\rm mod~}4)$, then $G(m,p,r)$-codes over finite fields of characteristic $p$ are asymptotically good.
 \end{coro}
 \begin{proof}
  This is a consequence of Theorem \ref{main} with $s=q=p$ 
  and of Corollary \ref{asym}
  for prime powers.
 \end{proof}
\section{Conclusion and Open Problem}\label{conclusion}
In this note, we have shown that left ideals in the group ring of a metacyclic group form, for certain values of the parameters, an asymptotically good family of codes.
The main open problem would be to extend this result to two-sided ideals. This would
allow to show, by the combinatorial equivalence derived in \cite{SL}, that abelian group codes are asymptotically good. Unfortunately, the conditions obtained in Remark \ref{rmk-abelian} seem to suggest that $s$-circulant metacyclic codes are not the right ones to be considered in this case.

\end{document}